\documentclass[conference]{IEEEtran}

\usepackage{amsmath}
\usepackage{amssymb}
\usepackage{algorithm}
\usepackage{algorithmicx}
\usepackage{algpseudocode}
\usepackage{subfigure}
\usepackage{acronym}
\usepackage{todonotes}
\usepackage[draft]{hyperref}
\usepackage{tabularx}
\usepackage[keeplastbox]{flushend} 
\usepackage{hhline}
\usepackage{paralist} 
\usepackage{amsmath}
\usepackage{amsthm}

\newtheorem{proposition}{Proposition}
\makeatletter
\renewcommand{\ALG@beginalgorithmic}{\scriptsize}
\makeatother

\algtext*{EndLoop}
\algtext*{EndFor}
\algtext*{EndIf}
\algtext*{EndWhile}

\newcommand{\name}{WiPLoc}
\newcommand{\mono}{Mobile Node}
\newcommand{\charbe}{Anchor Node}
\newcommand{\anbe}{WPA}

\acrodef{AC}{Alternating Current}
\acrodef{ESS}{energy storage system}
\acrodef{RF}{Radio Frequency}
\acrodef{RFID}{Radio Frequency Identification}
\acrodef{USB}{Universal Serial Bus}
\acrodef{WATS}{Wireless autonomous transducer systems}
\acrodef{WiFi}{wireless local area network}
\acrodef{WPT}{Wireless Power Transfer}
\acrodef{BLE}{Bluetooth Low Energy}
\acrodef{FEC}{Forward Error Correcting}
\acrodef{EIRP}{Effective Isotropic Radiated Power}
\acrodef{RSSI}{Received Signal Strength Indicator}
\acrodef{WFI}{Wait For Interrupt}
\acrodef{PRR}{Packet Reception Rate}

\begin{document}

\title{\name: Perpetual Indoor Localization with RF Wireless Power Transfer}
\author{
\IEEEauthorblockN{Qingzhi Liu, Wieger IJntema, Anass Drif, Przemys{\l}aw Pawe{\l}czak, and Marco Zuniga}
\IEEEauthorblockA{Department of Electrical Engineering, Mathematics and Computer Science\\ Delft University of Technology, Mekelweg 4, 2628 CD Delft, The Netherlands\\ Email: \{q.liu-1, p.pawelczak, m.a.zunigazamalloa\}@tudelft.nl, \{w.ijntema, a.drif\}@student.tudelft.nl}
%
}

\maketitle

\begin{abstract}
Indoor localization is a cornerstone of mobile services. Until now most of the research effort has focused on achieving sub-meter localization accuracy because many mobile applications depend on precise localization measurements. In some scenarios, however, it is required to trade location accuracy for system maintainability. 
For example, in large-scale deployments of indoor networks, such as item-monitoring in smart buildings, attaining room-level localization accuracy may be sufficient, but replacing the batteries of the devices used for localization could lead to high operational costs. As indoor localization systems grow in popularity it will be important to provide them with full energy autonomy. 
To tackle this problem we propose \name: an indoor localization system aimed at operating perpetually without batteries. 
Our contributions are twofold. First, we propose a novel localization method that exploits capture effect and orthogonal codes to operate at energy levels that are low enough to operate within the energy budget provided by long-range wireless power transmission. 
Second, we implement \name~using off-the-shelf components and test it extensively in a laboratory environment. Our test results show that with \name~one wireless charger per (16\,m$^{\text{2}}$) room can enable perpetual lifetime operation of mobile objects requiring localization with an average accuracy of almost 90\%. 
\end{abstract}

\section{Introduction}
\label{sec:introduction}

Indoor localization is a topic that has been investigated for more than a decade~\cite{gu:cst:2009,lymberopoulos:ipsn:2015}. Yet, no single system exists that is widely adopted as the de facto localization standard. 
From a research perspective the goal is to obtain a solution that is general, simple and accurate~\cite{lymberopoulos:ipsn:2015}. But as the number of devices in the Internet of Things grow, a new set of challenges are appearing: (i) maintainability, (ii) low operational costs, and more importantly, (iii) energy autonomy.

\subsection{Motivation: Indoor Localization with Wireless Power}

In large-scale indoor localization scenarios~\cite{lazik2015alps} the cost of replacing the batteries of thousands of anchor nodes (devices sending location information) and mobile nodes (devices to be localized) is high.

\textbf{Example:} Amsterdam Schiphol Airport has roughy 10 000 Bluetooth Low Energy (BLE) beacons deployed to provide navigation services. In order to maintain the operation of all BLE beacons, battery status monitoring mechanism of the localization beacons must be implemented. These include measuring beacon signal strength by crossing the entire airport area~\cite{moers:personal} or beacon signal strength crowdsourcing.
 
The research community has recognized this autonomous energy challenge and therefore the area of \ac{WPT} starts to gain momentum~\cite{xie:2012:wcm, xiao2014wireless}. In WPT systems a \emph{charger} radiates energy in the form of electromagnetic or mechanical waves to \emph{receivers} that harvest this energy. It would be thus valuable to combine the emerging area of WPT with the established area of indoor localization to propose a novel positioning system. 

\textbf{Challenges of RFID localization:} The idea of battery-less localization is not new. 
\ac{RFID} has been extensively researched for this purpose. 
Unfortunately, RFID-based localization has some inherent limitations. 
First, most \ac{RFID} localization approaches, e.g.~\cite{saab2011standalone}, require pre-deployed anchor tags,  
where portable \ac{RFID} readers estimate their position by detecting nearby backscatter signals from tags. 
Although the \ac{RFID} tags are batteryless, the mobile reader requires a lot of energy during tag scanning~\cite[Table I]{liu:2014:tud}, and due to the fast attenuation of backscatter signals, the density of \ac{RFID} tags must be high. 
Second, \ac{RFID} technology is also used in tracking systems~\cite{yang2014tagoram}, where mobile nodes carry a passive tag and static readers are used to track their location. In principle, this system is similar to ours, with WPT chargers playing the role of RFID readers. But \name~has the added advantage of having both the monitoring system and the tag itself being aware of the location. With RFID tracking, only the system knows the location of the tags, but the tags themselves are not aware of their own location. 

\subsection{Wireless Powered Indoor Localization: Research Challenge}

Localization and WPT are well researched topics on their own, but using WPT for localization entails a substantial challenge.
The problem is that WPT provides amounts of power that are too small for the operation of most radio-based localization systems. For example, experimenting with TX91501 power transmitter~\cite{tx91501_p2110}, due to the exponential decay of signal strength, the harvested power is 0.79\,mW at 3\,m. However, highly energy-efficient radios, such as BLE nodes, consume around 25\,mW in receiving mode. 
This small amount of power is insufficient for not only receiving packets from many anchor nodes but also to synchronize the operation of the localization system. 

\textbf{Research Question:} Based on the observation above we define the research problem as: \textit{given the limited harvested energy from WPT, how should a system manage the indoor localization process to achieve continuous and perpetual localization?}

\subsection{Our Contributions}
\label{sec:contribution}

Considering the research question described above we propose a unified set of solutions for wirelessly powered indoor localization. Namely:

\textbf{Contribution 1:} To minimize the radio transmission and receiving time of localization and meet the limited harvested energy supply, we propose a novel localization method where all anchors transmit localization messages \emph{simultaneously} and the induced collisions are resolved via orthogonal codes. The key advantage of this approach is that the radios of all nodes, including anchor nodes and nodes requiring localization, are active only for the duration of a single message transmission. 

\textbf{Contribution 2:} We implement and evaluate an operational system 
using off-the-shelf BLE motes~\cite{nrf51822} and WPT chargers and harvesters~\cite{tx91501_p2110}. Based on systematic experiments in an office environment, we demonstrate that  \name~can achieve perpetual indoor localization with room-level ($\approx$16\,m$^{\text{2}}$) and cell-level ($\approx$\,4\,m$^{\text{2}}$) accuracies of approximately 90\% and 70\%, respectively. To the best of our knowledge, \name~is the first system that successfully achieves room-level localization using the energy of RF based WPT. 

The rest of the paper is organized as follows. 
The related work is discussed in Section~\ref{sec:realted_work}. Our basic battery-less localization method achieving room-level accuracy, denoted as \name, is presented in Section~\ref{sec:intro_wiploc}, while its experimental evaluation is presented in Section~\ref{sec:results_wiploc_basic}.
Approaches to further save power and increase localization accuracy, denoted as \name++, are presented in Section~\ref{sec:wiploc_extension}, with its detailed evaluation presented in Section~\ref{sec:wiplocplus_implementation_evaluation}.
Finally, we present our conclusions in Section\textcolor{black}{~\ref{sec:conclusion}}.

\section{Related Work}
\label{sec:realted_work}

\subsection{Localization with Wireless Power Transfer}

The field of (indoor) localization has been researched for years~\cite{gu:cst:2009,lymberopoulos:ipsn:2015}. 
Interestingly engough to the best of our knowledge, we are not aware of any localization technique that uses \ac{WPT} except for TOC~\cite{shu:infocom:2014}. TOC obtains location information based on the time of charge provided by mobile chargers to the static nodes being localized. Unfortunately \begin{inparaenum}[(i)]\item TOC requires frequent position changes of mobile charger to obtain reasonable location accuracy and \item has been tested in outdoors only. \end{inparaenum} TOC belongs to a static receiver/mobile charger \ac{WPT} network type. Following the categorization of~\cite[Sec. III-A]{liu:2014:tud}, none of three remaining categories have been applied for localization (either indoor or outdoor). Specifically, localization infrastructure where static charger (not obstructing the area of localization) and mobile receiver being localized, which is the most desired.

\subsection{Localization with RFID}

Localization based on \ac{RFID} technology is developing rapidly in recent years. 
These approaches can be classified into three categories~\cite[Sec.III]{ni2011rfid}.
In the first category, RFID reader-based localization system, e.g.~\cite{zhu2014fault, saab2011standalone}, allocates RFID reader in the object requiring localization to detect the pre-deployed anchor tags nearby. 
Although the deployment and maintenance cost of RFID tags are low, the localization lifetime depends on the limited battery of mobile readers.
In the second category, RFID tag-based localization system, e.g. \cite{yang2014tagoram, ni2004landmarc}, tracks the location of object attached with \ac{RFID} tag by pre-deployed readers. 
The advantage of this approach is that the lifetime of tags is unlimited. 
However, the localization range is limited by the density of readers. 
Also, the tags cannot be localized once they are outside the range of readers. 
In the last category, \ac{RFID} device-free localization system, e.g. \cite{liu2012mining}, detects the position of target wearing no additional localization devices. 
The idea is to find the target location by detecting and comparing the change of \ac{RFID} signal in the environment. 
Tag-based and device-free localization are both passive localization/tracking, which do not fall into the research area of this paper. 

\section{\name: Wirelessly-Powered Localization}
\label{sec:intro_wiploc}

We propose a new indoor localization method for a freely moving device that trades-off accuracy for power consumption. The goal is to reduce the energy consumption of the localization system to a level allowing to power the localization infrastructure wirelessly over large distances.

\textbf{Selection of WPT technology.}
While many \ac{WPT} techniques exist~\cite[Table 1]{xie:2012:wcm} we chose the one based on \ac{RF} for two reasons. First, RF signals can serve the dual purpose of providing localization and energy. This simplifies the design, operation and cost of anchors and tags compared to systems using different signals for localization and energy transfer, such as sound or magnetic resonance. Secondly, it is the most promising \ac{WPT} technology as it allows for long-range power transfer even with small receiver antennas, which is challenging with other \ac{WPT} technologies such as induction-based \ac{WPT}. 

\textbf{How \name\, enables WPT-based localization.}
The main idea behind \name's low power consumption is to exploit synchronous packet transmissions to reduce the radio activity time and, in-turn, collision resolution through improved packet capture. 
In the subsequent sections we will introduce two main \name~ processes: (i) localization (Section~\ref{sec:challenge1}) and (ii) wireless energy supply (Section~\ref{sec:wptsystem_integration}), in detail.

\subsection{\name: Localization Protocol}
\label{sec:challenge1}

To get a clear understanding of how \name~ localization works we introduce all localization building blocks in detail.

\subsubsection{Deployment Area}

The system is designed for indoor use. \name's aim is to find a location of the moving object within strictly defined localization areas, i.e. rooms of an office environment. 

\subsubsection{Components}

The \name~ system consist of two building blocks: (i) Anchor Nodes and (ii) Mobile Nodes that want to be localized.
\begin{itemize}
\item {\textbf{Anchor Node:}} This node is deployed as static in a room and placed at specific locations that maximize signal reception by the Mobile Node in that room. The Anchor Node is pre-programmed with a unique ID that correlates with that specific room. The Anchor Node is constantly powered via a cable and is always in a receiving mode. For deployment simplicity each room has only one Anchor Node.
\item {\textbf{Mobile Node:}} This node can move between rooms and is powered by some form of wireless energy. The Mobile Node is constantly in sleeping mode and only wakes up after a pre-defined time after which it goes back to sleep. It has a pre-defined table with all anchor IDs to correlate the Anchors Nodes with a location (room).
\end{itemize}

\subsubsection{Localisation Algorithm}

The network is consisting of Anchor Nodes and Mobile Nodes, where the Mobile Nodes are localized by receiving the ID of the strongest anchor. \name~ localization belongs therefore to the proximity-based methods of indoor localization systems~\cite[Sec. II-C]{liu:smc:2007}. 

Localization methods that use packet radio are usually asynchronous to avoid collisions among packets sent by other anchors. In \name~however, all anchors send their packets at the same time, enforcing packet collisions. 
The mobile node leverages the capture effect~\cite[Sec. II-A]{arnbak:jsac:1987} to decode the strongest signal and assigns its location to that anchor. 
The key advantage of this method is its energy efficiency: Anchors and Mobile Nodes in the \name~network only need to be active for a single packet transmission and reception time slot.
A separate discussion is needed on packet synchronization, collision resolution and error correction.

\paragraph{Packet Synchronisation}
\label{sec:packet_synchronisation}

The packets from the anchors need to arrive  at the Mobile Node at the same time. 
For example, in our protocol each packet consists of a preamble of one byte, a payload and a CRC of the whole packet. 
To leverage the capture effect the packets should arrive within each other's preamble at the Mobile Node. 
To achieve this synchronization the Mobile Node broadcasts a \texttt{location-request} packet.
This is a synchronization packet that instructs all receiving anchors to immediately respond with their ID encoded in a payload. 
Leveraging the capture effect alone however has limitations~\cite[Sec. IV]{velze:percom:2013}. 
For the capture effect to work the strongest signal needs a certain minimum SINR. 
If this requirement is not satisfied, packets will collide and the anchor ID will not be retrieved.

\paragraph{Orthogonal Spreading Codes}
\label{sec:orthogonal_codes}

To overcome this limitation the Anchor Node ID is encoded with an orthogonal code to increase inter-packet distinction.
Each bit of the Anchor Node ID is multiplied by an orthogonal code unique for each anchor.
The encoded Anchor Node ID is then send in the payload of a packet.
The decoding process at each mobile node is an XOR operation between the payload of received packet with a list of orthogonal codes. In this paper a Hadamard matrix of size $k$ is used to generate the codes\footnote{Any other method can be used as long as the codes all have zero cross-correlation with each other.}, i.e.,

\begin{equation}
H_{2^k} = \begin{bmatrix} H_{2^{k-1}} & H_{2^{k-1}}\\ H_{2^{k-1}} & -H_{2^{k-1}}\end{bmatrix} = H_2\otimes H_{2^{k-1}},
\end{equation}
where $H_2 = \begin{bmatrix} 1 & 1 \\ 1 & -1 \end{bmatrix}$, $2 \leq k \in \mathbb{N}$
and $\otimes$ is the Kronecker product.

\paragraph{FEC layer}

Although orthogonal codes have a build-in tolerance for bit errors they are unable to decode them all.
For this reason the Anchor Node ID is first encoded with a \ac{FEC} code, before it is multiplied by the orthogonal codes.
The \ac{FEC} layer is constructed of maximum minimum Hamming distance codes~\cite{macdonald1960design} i.e. codes having equal Hamming distance to each other. 
The decoding of the FEC layer uses minimum distance decoding.

\subsubsection{Illustration---Eliminating Localization Dead Zone}
\label{sec:dead_zone_experiment}
To verify that the orthogonal codes are working correctly the following experiment was performed.

\paragraph{Experiment Hardware}
\label{sec:ble_hardware}

For convenience we select the nRF51822 SoC with a ARM Cortex M0 from Nordic Semiconductor with BLE support~\cite{nrf51822} as hardware platform to test the localization protocol.
For \name~however, we do not use the BLE protocol stack but use the radio peripheral of the nRF51822 and introduce our own packet and communication protocol instead. We refer to~\cite{source_code_wiploc} for the source code of the implementation.

\begin{itemize}
\item {\textbf{Anchor Node:}} The Smart Beacon Kit form Nordic Semiconductor~\cite{nrf51822} is used. This module has a coin size form factor with a PCB integrated antenna.
\item {\textbf{Mobile Node:}} The Nordic Semicoductor PCA10005~\cite{nrf51822} is used. It has an SMA connector with a connected quarter-wave helical monopole antenna of 1.6\,dBi gain\footnote{In the rest of the papers we will refer to both devices as nRF51822.}.
\end{itemize}

\paragraph{Experiment Setup}
Two anchors are placed two meters apart and the mobile node is placed at 20 points in a straight line between the anchors with each measurement point separated 10\,cm from each other. 
The transmission power of the Anchor Nodes are set to 0\,dBm. The mobile node stays at each point for around 30\,s and sends localization request every second.
First the localization experiment was performed without orthogonal codes.
After that the same experiment was performed with orthogonal codes.

\paragraph{Experiment Results}
In Fig.~\ref{exp:localization_before} we observe that in-between anchors, we obtain a \emph{dead zone}, i.e. area of no reception because in this region the SINR is insufficient to receive a correct packet and the packets collides which causes the packets to deform.
Fig.~\ref{exp:localization_after} demonstrates that the dead zone is eliminated with orthogonal codes. Furthermore, we observe that multiple anchor IDs can be decoded from one packet. 
We call this phenomenon \emph{multi-packet reception}.
The reason for the larger coverage of Anchor 1 in Fig.~\ref{exp:localization_after} is due to use of antenna with a different coverage pattern in that experiment.

\begin{figure}%
	\centering
	\subfigure[Without orthogonal codes]{
		\includegraphics[width=0.45\columnwidth]{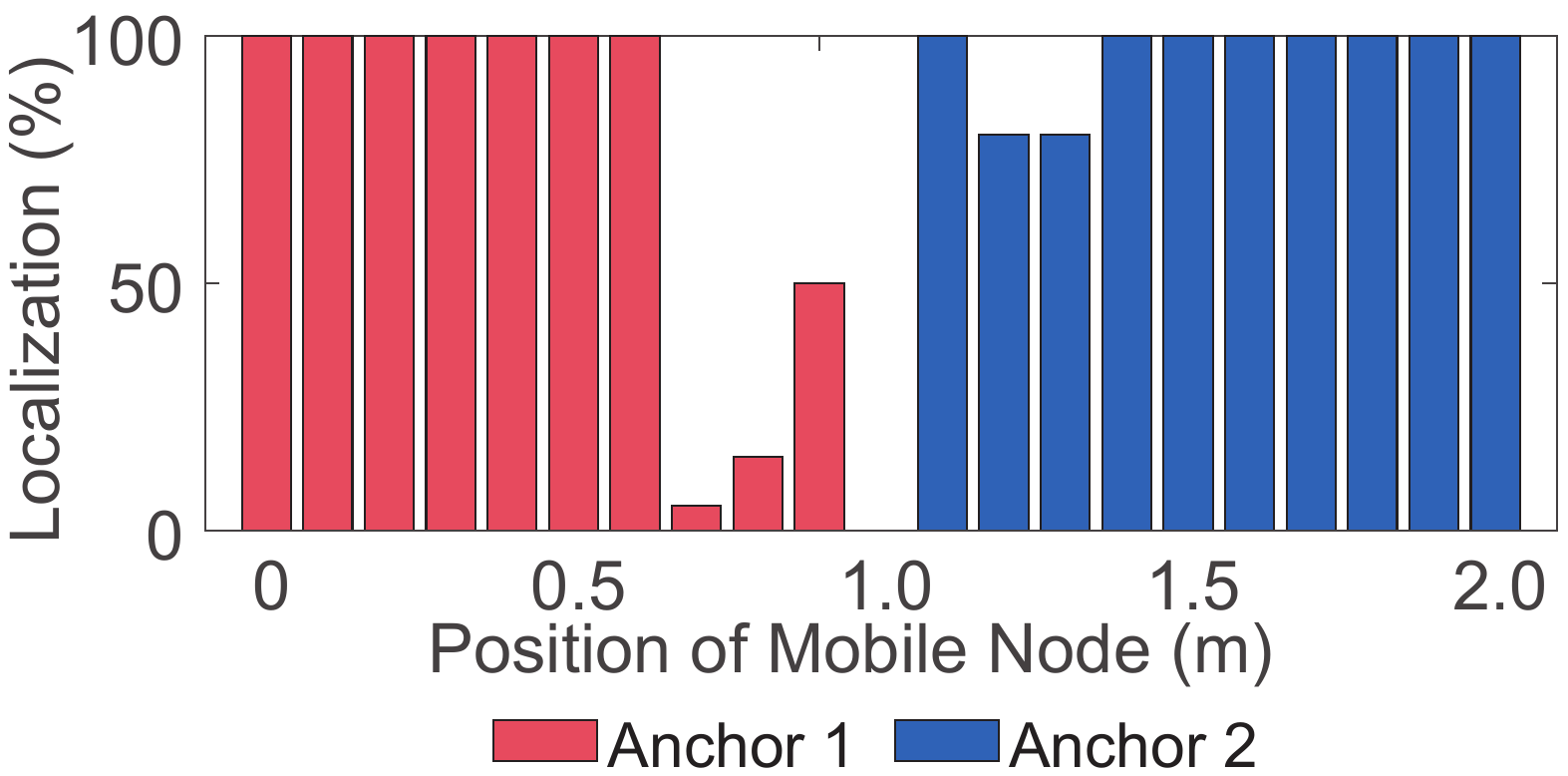}
		\label{exp:localization_before}
	}
	\subfigure[With orthogonal codes]{
		\includegraphics[width=0.45\columnwidth]{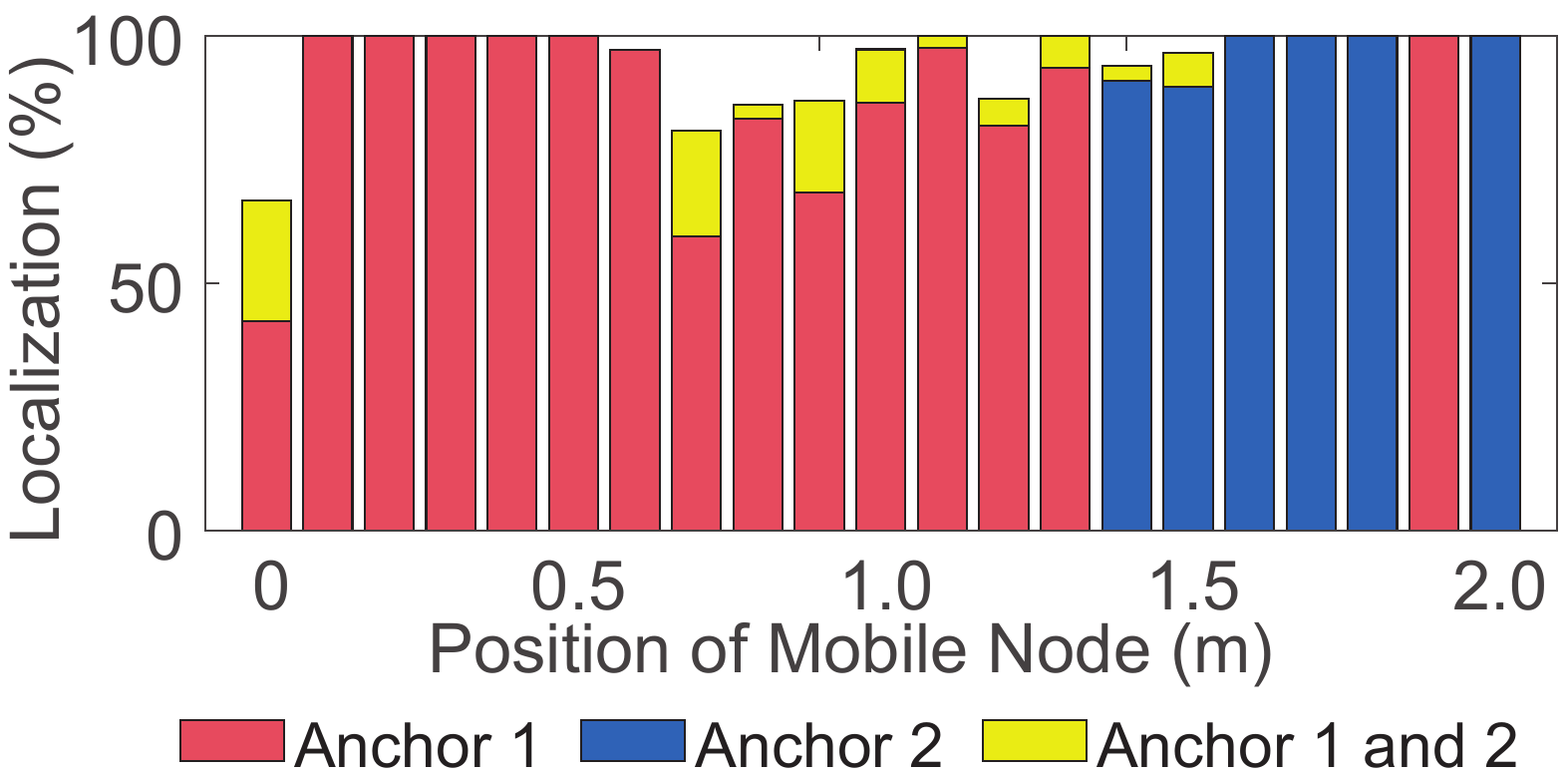}
		\label{exp:localization_after}
	}
	\caption{Localization accuracy experiment with two anchors placed two meters apart. We demonstrate the elimination of the packet reception \emph{dead zone} with orthogonal codes, see Section \ref{sec:dead_zone_experiment} for details of the experiment setup.}
	\label{fig:localization_before_after}
\end{figure}

\subsection{\name: Wireless Energy Supply}
\label{sec:wptsystem_integration}
Having low power features of \name~implemented, we are ready to extend \name~with wireless power localization features. As in Section~\ref{sec:challenge1} we describe all localization blocks with WPT enabled.

\subsubsection{Components}
\label{sec:wpt_components}
The \name~components described in Section~\ref{sec:challenge1} are in this section extended with WPT functionality.

\begin{itemize}
\item{\textbf{Anchor Node:}} As the anchor nodes are located at static and central in a room we combine them with a power transmitter. This transmitter should provide full coverage for one room.
\item{\textbf{Mobile node:}} This node is combined with a power harvester that harvests RF power form the RF power transmitter. Note that the WPT channel is different than the communication channel.
\end{itemize}

\subsubsection{Software Implementation}
\label{sec:software_implementation}

The localization algorithm is the same as in Section~\ref{sec:challenge1} however to enable it with WPT the following steps have been taken.
To keep the power consumption at a minimum all the peripherals of the Mobile Node are turned off except for one hardware timer which is set to generate an interrupt\footnote{We again refer to~\cite{source_code_wiploc} for the source code of the implementation.} at each period $t_m$. 
The CPU of the nRF51822 is most of the time in \emph{\ac{WFI}} state, which we denote as the \emph{sleep state}. 
When the interrupt is executed the nRF51822 is woken up and begins a localization round.
The localization round starts broadcasting a \texttt{location-request}. Next, the radio directly switches to receiving mode. 
If a valid location packet is received (meaning the CRC is correct) decoding the anchor ID is trivial. 
On the other hand, if the received packet is corrupted the orthogonal codes inside ensure that the anchor ID can still be decoded from the packets.
After this the Mobile Node goes back to sleep.
The complete program flow is depicted in Algorithm~\ref{alg:room_level}. 

\begin{algorithm}[t] 
\caption{\small Location protocol at the Anchor and Mobile Node}
{\scriptsize Anchor Node:}
\begin{algorithmic}[1]
\Loop \label{alg1:anchor_room}
	\If{\texttt{location-request} packet is received}
		\State Send \texttt{location-reply} packet back
	\EndIf
\EndLoop 
\end{algorithmic}
{\scriptsize Mobile Node:}
\begin{algorithmic}[1]
\Loop \label{alg1:mono_room}
	\State Sleep
	\If{timer $\geq t_m$} \Comment{See Section~\ref{sec:software_implementation}}
		\State Broadcast \texttt{location-request}
		\If{\texttt{location-reply} packet received}
			\State Decode and calculate location
		\EndIf		
	\EndIf
\EndLoop 
\end{algorithmic}
\label{alg:room_level}
\end{algorithm}

There are two types of packets sent by the \name~protocol. Both types have a fixed payload length of 30 bytes. 
Next section will elaborate on how the payload of the packets is constructed for each type.

\paragraph{Location Request Packet}

The\,\texttt{location-request} packet contains in the first two bytes of the payload the group ID of anchors.
The group ID of all anchors is the arbitrary integer $j$.
The anchors are then programmed to accept all packets that contain this integer $j$ in the first two bytes. 
The packet send back by the Anchor to the Mobile Node is a \texttt{location-reply} packet.

\paragraph{Location Reply Packet}

The \texttt{location-reply} packet contains in the payload the anchor ID encoded by the FEC layer which is then encoded by the orthogonal codes and this is stored in the payload of the packet.

\paragraph{Location Reply Encoding Process}
All the Anchor Nodes have an array $\mathbf{C}$, $|\mathbf{C}|=N$ with FEC codes having an equal Hamming distance $d$ to each other. 
The anchor ID is then FEC encoded by replacing the anchor ID with a $n$-th FEC code from the array.
The orthogonal codes are generated using the method described in Section~\ref{sec:orthogonal_codes}.
For the generated matrix each $-1$ symbol is replaced with a $0$ and each row of the matrix represents a binary spreading code.
Finally, every bit of the FEC code is represented by the $n$-th orthogonal code from the array. 
If the bit is zero, the bitwise NOT of the orthogonal code is used\footnote{The size of the orthogonal codes is 16 bits and this results in a total message length of 30 bytes that fits in the payload of the radio packet.}.

\paragraph{Location Reply Decoding Process}
As stated in Section~\ref{sec:challenge1} also the Mobile Node has an array of all orthogonal and \ac{FEC} codes used by the anchors.
For every entry in the orthogonal code array the Mobile Node tries to decode the packet. 
When a candidate anchor ID is found, the decoded code is compared to the correlating FEC code of the candidate ID. 
When the Hamming distance $d_c$ of the candidate FEC code compared to a code from the FEC array follows $d_c < \frac{d}{2}$ we assume that the candidate ID is the correct one.
The decoding stops when the last orthogonal code in the array is used to decode a packet. 
The source code accompanying this paper is available at~\cite{source_code_wiploc}.

\subsubsection{Hardware Implementation}
\label{sec:hardware_implementation}
The \name~localization components are connected as follows, see Fig.~\ref{fig:schematic_complete}:

\begin{itemize}
\item{\textbf{Anchor Node:}} This node is combined with the Powercast TX91501 Powercaster transmitter~\cite{tx91501_p2110}. 
It has an \ac{EIRP} of 3\,W and operates at a \ac{RF} center frequency of 915\,MHz. 
There is no signal connecting the powercaster to the nRF51882. 
The Powercaster is always active and is always sending \ac{RF} energy into the environment. In the rest of the paper we will refer to this component as \emph{powercaster}.
 
\item{\textbf{Mobile Node:}} The Powercast P2110B power harverster~\cite{tx91501_p2110} is used as power supply for the nRF51822. 
It is a development PCB with a SMA connector for connecting an antenna for harvesting power. 
Two antennas can be selected: (i) a vertical polarized omni-directional dipole antenna with 1.0\,dBi gain and (ii) a vertical polarized patch antenna with a 6.1\,dBi gain. 
We will refer to this component as the \emph{harvester}. 
\end{itemize}

\begin{figure}
	\centering
	\subfigure[Mobile Node implementation]{
		\includegraphics[height=2.50cm]{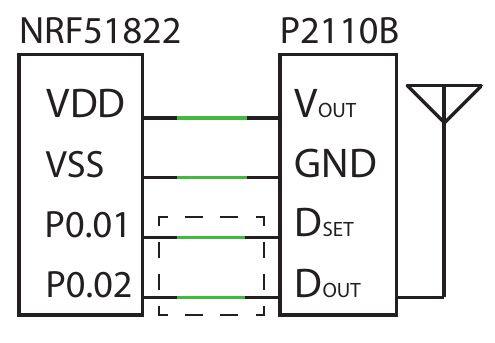}
		\label{fig:connections_nrf5188_powercast}
	}\subfigure[Anchor Node implementation]{
		\includegraphics[height=2.50cm]{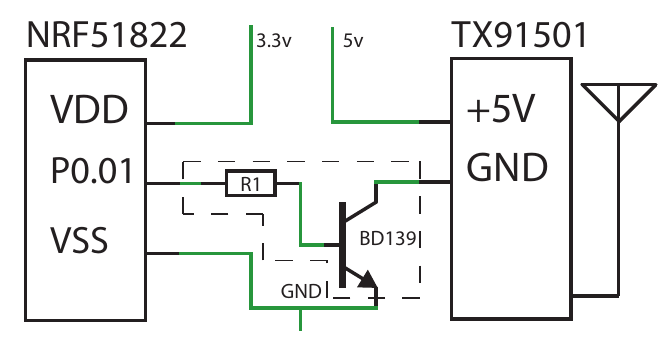}
		\label{fig:charbe_schematic}
	}
	\caption{Implementation of \name~components: (a) \mono: Nordic Semiconductors nRF51822 SoC~\cite{nrf51822} connected to Powercast energy harvester P2110B~\cite{tx91501_p2110}; and (b) \charbe: nRF51822 connected to Powercast power transmitter~\cite{tx91501_p2110}. Note that in case of nRF51822 only the relevant pin connections are shown as other pins are left as not connected. The connections inside the dotted line are only applicable for \name++, refer to Section~\ref{sec:wiploc_extension}. The value for $R_1=330$\,$\Omega$.}
	\label{fig:schematic_complete}
\end{figure}

\section{\name: Experiment results}
\label{sec:results_wiploc_basic}

From~\cite{cirstea:2013:aece} we know that the lower bound for harvested energy from a distance of four meters from the powercaster is --8\,dBm. In order to operate within the powercaster range in our setup the average power consumption of the Mobile Node should be below 0.16\,mW.

To verify the energy efficiency of the \name~components the nRF51822 is measured with a Power Monitor~\cite{monsoon_website}.
The nRF51822 has three power states: \emph{sleep} (WFI), \emph{transmitting} (TX) and \emph{receiving} (RX). 
The localization period $t_m = 1.0$\,s and the transmit power is set to 4\,dBm.
A  localization round starts with a TX followed by RX and WFI.
We measured the energy consumption of each state separately, repeated each measurement ten times and averaged them.
Table~\ref{tab:power_mobile_node} presents the power consumption measurements.
The measurements show that the average power consumption is 0.20\,mW which is very close to the requirement of 0.16\,mW. We are thus ready to implement \name.

\begin{table}
	{\centering
	\caption{Consumed power in each fundamental state of the \name~localizarion protocol}	
	\label{tab:power_mobile_node}
	\begin{tabular}{ c | c|c | c|c | c|c p{\columnwidth}}
		State & \multicolumn{2}{c|}{Power (mW)} & \multicolumn{2}{c|}{Time (ms)} & \multicolumn{2}{c}{Energy ($\mu$J)}\\
		\hline
		\hline
		Transmitting& 35.88&35.88 & 0.80&0.80 & 28.7&28.7\\
		\hline
		Receiving & 20.17&26.05 & 0.60&8.29 & 12.1&216.1\\
		\hline
		ADC & --- &1.69 & --- &0.65 & --- &1.10\\
		\hline
		WFI & 0.15&0.14 & 998.60&925.9 & 149.8&138.9\\
		\hline\hline
		Average& 0.19&0.49 & 1000.0&1000.0  & 190.6&493.4 \\
	\end{tabular}
	}\\ \\
	Note: The left column denotes the power consumption of the Mobile Node, the right column denotes the WPA (See Section~\ref{sec:wiploc_extension}).
\end{table}

\subsection{Experiment Setup}

Each location is divided in four of two by two meters cells. 
At the center of the cell is an test location for the mobile node. Each room has four test locations. Every device is placed 1.0\,m above the floor and they are all in line of sight from each other.

For every experiment that is done the following yields. 
The mobile node is placed at every testing location in room one and two and the corridor. 
On every test location the mobile node initiates 50 localization rounds. The localization period is set to 1.0\,s and the transmit power of the Anchor and Mobile Node is set at the maximum of 4\,dBm.

\subsubsection{Experiment Scenarios}
The following three experiments were performed. 

\textbf{Experiment 1:} One Anchor Node is placed in the center of one room and the Mobile Node is placed at every testing location in room one and two and the corridor.  
As there is only one Anchor Node in the area we consider all packets decoded with this anchor ID correctly localized.

\textbf{Experiment 2:} Two Anchor Nodes are deployed, each in the center of room one and two.
Two Voronoi cells around the Anchor Nodes are defined and if the Anchor Node is localized in the Voronoi cell that corresponds to the Anchor ID that is decoded we consider it as correctly localized.

\textbf{Experiment 3:} Three Anchor Nodes are deployed, each in the center of room one,  two and the corridor. 
The Mobile Node is correct if the localization result is in the correct room. 
We make use of the fact that the walls attenuate the RF signals from the nRF51822 and that the signal range will adjust to the room layout accordingly.

\subsubsection{Data Acquisition}
As the Mobile Node is powered wirelessly interfacing with the nRF51822 would consume power which we then can not be used for localization. 
We overcome this problem by using a BLE USB dongle as a sniffer~\cite{nrf51822}. 
This sniffer monitors all packets sent by the Mobile Node.
In our experiment the result of every localization round is send in the next localization round within the \texttt{location-request} packet.
The sniffer then receives the data and saves it on to a text file on the PC for further data processing.

\begin{figure}  
\centering
\subfigure[Top view of \name~deployment in the office environment]{\includegraphics[height=0.45\columnwidth]{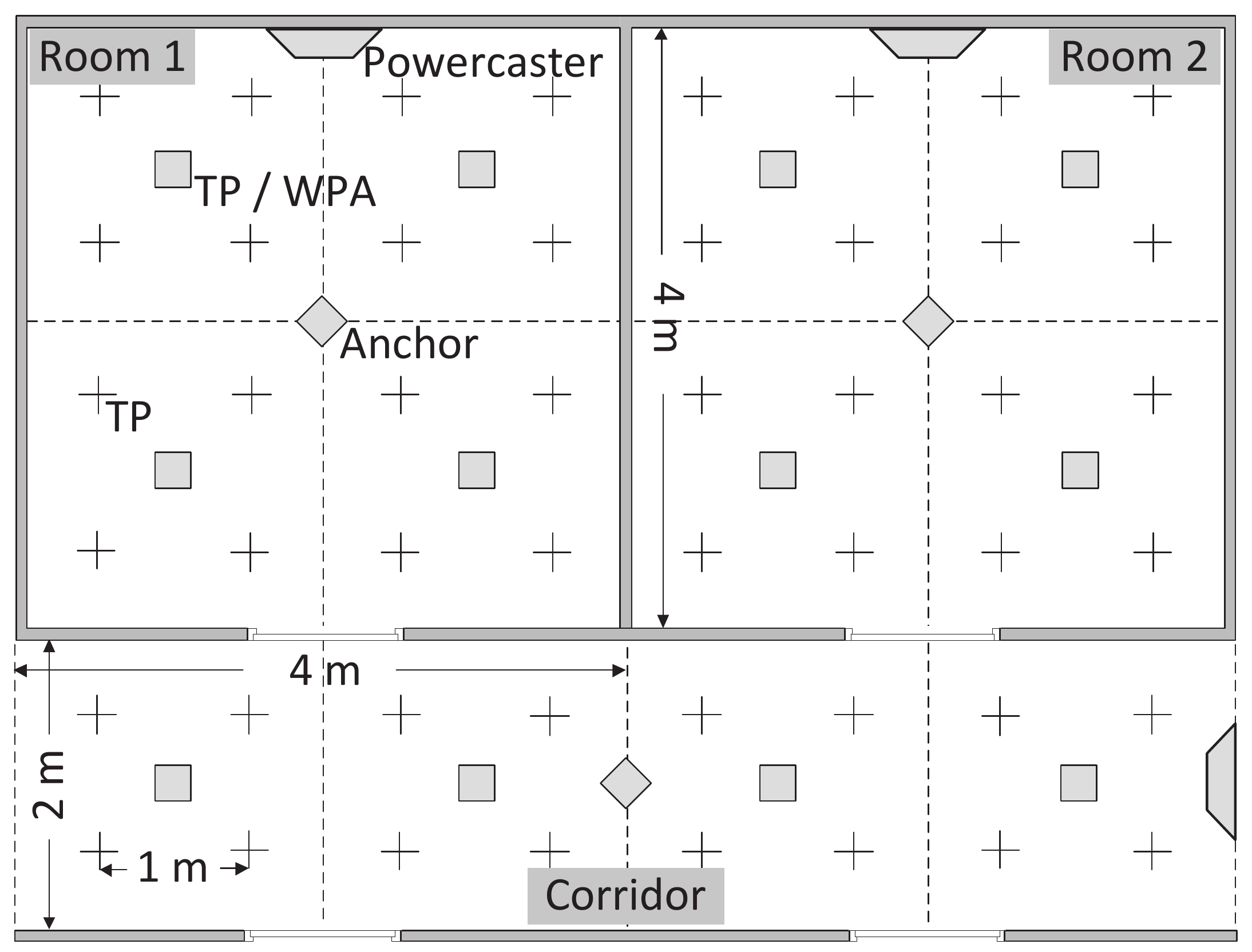}\label{fig:experiment_setup}}
\subfigure[A picture of corridor \name~experiment setup]{\includegraphics[height=0.45\columnwidth]{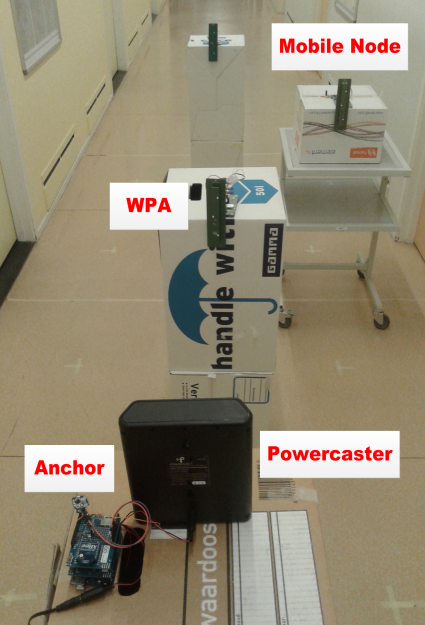}\label{fig:photo_setup}}
\caption{Experiment setup: \textbf{TP/WPA}: The places marked as "$\square$" are used as the testing positions (TP) of {\mono} in the room-level localization experiment of Section~\ref{sec:results_wiploc_basic} and the deployment positions of {\anbe} and in the cell-level localization experiment of Section~\ref{sec:wiploc_extension}, respectively. \textbf{TP}: places marked as "$+$" are the testing positions (TP) of {\mono} in the cell-level localization experiment of Section~\ref{sec:wiploc_extension}.}
\label{fig:experiment_cell_level}
\end{figure}

\subsection{Experiment Results}

For the evaluation of \name~we introduce two metrics: \emph{accuracy} and \emph{\ac{PRR}}. 
We define accuracy as follows. When the Mobile Node is localized in the room where it is currently located we count the localization as successful. When the Mobile Lode is localized to another room then where it is currently located it is counted as unsuccessful.
PRR is defined as the number of \texttt{location-request} packets send by the mobile node divided by the number of received \texttt{location-reply} packets.

At every test location 50 localization rounds were performed.
For each test location \ac{PRR} and the accuracy is computed and averaged.
The results are shown in Table~\ref{tab:result1} where the accuracy is normalized to the PRR. 
We observe that if the number of anchors deployed increases the \ac{PRR} decreases. 
The accuracy also decreases when the number of anchors increases.
The \ac{PRR} decreasing is mostly due the fact of packet collisions.
Nevertheless, even in the worst case (Experiment 3) we demonstrate that we achieve (i) extremely low power consumption for localization using collision packets, and (ii) accurate room-level localization for {\mono} using only WPT energy.

\begin{table}
	\centering
	\caption{\name~localization experiment result}
	\label{tab:result1}
	\begin{tabular}{ c | c | c  p{\columnwidth}}
		 & PRR (\%) & Accuracy (\%)\\
		\hline
		\hline
		Experiment 1& 100 & 100 \\

		Experiment 2& 99.3 & 95.5 \\

		Experiment 3& 89.6 & 84.6 \\
	\end{tabular}
\end{table}

\section{\name++: Extending \name~to Cell-Level Localization Accuracy}
\label{sec:wiploc_extension}

So far we have demonstrated how \name~allows us to accurately localize items per room. 
The question is how to improve localization accuracy from room-level to cell-level.

\subsection{\name: Challenge of Cell-level Localization}

\name~is designed to cope with localization at a room-level ($\approx$\,16\,${{\rm{m}}^2}$) accuracy but does not allow to improve the accuracy further down.
We shall describe specific problems of \name~ related to this functionality and propose our improvements, which we collectively shall denote as \name++.

\subsubsection*{\textbf{Problem 1}---Limited Energy for Synchronization}

Referring again to Table~\ref{tab:power_mobile_node}, we see the the power consumption of packet reception ($\approx$26\,mW) is above the limit of harvested power ($\approx$0.8\,mW at 3\,m). 
Synchronization is needed when operating with multiple Anchor Nodes this requires idle listening continuously, However this is not possible if {\charbe}s are required to operate under the harvested RF energy. 

\textbf{Solution to Problem 1:} 
We propose a semi-passive wakeup scheme to allow {\charbe}s listening to the synchronization signal only when there is a localization request. 
This solution will be described in Section~\ref{sec:passive_wakeup}.

\subsubsection*{\textbf{Problem 2}---Scalability of \name}

Normally, with more {\charbe}s we obtain higher localization accuracy. 
To prove this assumption, we deploy 2, 3 and 4 {\charbe}s, respectively, in room one and test the PRR and accuracy respectively. 
The deployment setup and localization testing positions are illustrated in Fig.~\ref{fig:experiment_setup} while the test results are shown in Fig.~\ref{pic:cell_level_performance_drop}.
The key observation is that the PRR and accuracy of cell-level localization decrease as the number of {\charbe}s increase.
This is mainly caused by the radio interference to the orthogonal code from multiple {\charbe}s. 
It means that interference caused by multiple {\charbe}s will limit the accuracy and scalability of \name~in dense Anchor Nodes deployments. 

\textbf{Solution to Problem 2:} 
We propose to restrict the number of {\charbe}s used for cell-level localization by using WPT Receive Signal Strength (RSS) based distance estimation, which will be presented in Section~\ref{sec:distance_estimation}.

\begin{figure}
 \centering
 \includegraphics[width=\columnwidth]{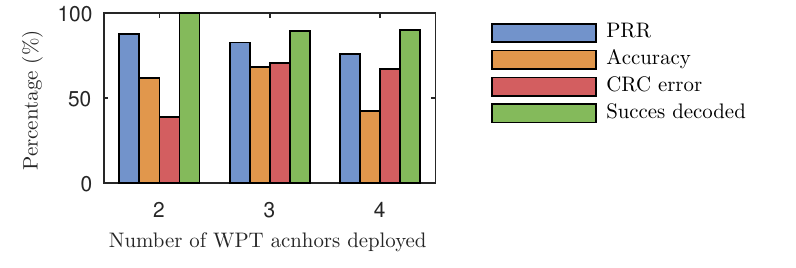}
 \caption{Cell-level localization as a function of number of Anchor Nodes deployed. The radio interference from multiple Anchor Nodes causes the increase of CRC error and the decrease of success decode rate.}
 \label{pic:cell_level_performance_drop}
\end{figure}

\subsection{ID Based Semi-Passive Wakeup}
\label{sec:passive_wakeup}

The collision based localization introduced in Section~\ref{sec:intro_wiploc} requests the synchronization of packet transmission and reception among multiple nodes. 
The {\charbe}, also responsible for charging, has a fixed power supply, \textit{ergo} enough energy to listen to the localization request and synchronization signal continuously. 
If we want to increase the number of {\charbe}s without a fixed power supply, they also need to be be powered by wireless power.

\subsubsection{Extra Localization Component}

Following from the above observation we introduce a new node aiding in localization.

\begin{itemize}
\item \textbf{Wirelessly-Powered Anchor Node (\textbf{{\anbe}}):} This node is the same as Anchor Node, however it operates purely based on harvested RF energy from  the {\charbe}.
\end{itemize}

Energy harvested from {\charbe} is not enough for {\anbe}s to listen to the synchronization signal continuously. 
Therefore, we propose ID based semi-passive wakeup approach to wakeup {\anbe}s from sleeping mode only when the localization request is sent from {\mono}s. 
The method works as follows. 

\subsubsection{Wakeup Process} 

As in {\name}, {\charbe}s are deployed one per room and constantly switched on for wireless charging. 
Using the harvesting energy from neighbor {\charbe}s, {\anbe}s periodically wakeup from sleep mode and perform the measurement of the voltage of the harvested power with Analog to Digital Converter (ADC) port. 
Based on our measurements (see again Table~\ref{tab:power_mobile_node}), the power consumption of ADC measurement ($\approx$\,1.5\,mW) is of the same magnitude as the harvested power (from $\approx$\,3.2\,mW at 1\,m to $\approx$\,0.79\,mW at 3\,m). 
Although the power consumption of ADC measurement is larger than harvested power at the distance of 3\,m. We only conduct one ADC measurement every period $t_c$ and are in sleeping mode the rest of time. This ensures that the average power consumption low enough.

Now, suppose that the {\mono} has initiated the room-level localization. 
Then {\mono} broadcasts a \texttt{localization request} to the Anchor Nodes.
The Anchor Node will send a \texttt{location reply} back and will send a passive wake up signal to all WPAs in the same room: switch off (the powercaster) charging for a short time and then switch back on again. 
Then the {\charbe} sends \texttt{sleep} commands to the {\anbe}s that are required to stay a sleep for the cell-level localization.
 
At the WPA, Once a voltage falling of the harvested power is measured on $D_{\text{out}}$, the {\anbe} wakes up from sleeping mode and starts listening to radio packets.
If the {\anbe} does receive a \texttt{sleep} it keeps listening for a \texttt{location-request} from the Mobile Node.

Meanwhile, the {\mono} received the \texttt{location-reply} from the Anchor Node and then waits for $t_c$ to broadcast the \texttt{location-request} to the WPAs.
Then {\mono} is localized based on the collision based localization approach as explained in Section~\ref{sec:intro_wiploc}. 
The whole process is described in Algorithm~\ref{alg:wiploc_localization_protocol}. 

It is worth noting that the wake up time of each \anbe~and \mono~have small differences. 
The main reason for that is the time for measuring the semi-passive wake up signal at each {\anbe} is not synchronized. 
To guarantee that all {\anbe}s can hear the localization synchronization signal from {\mono}, all {\anbe}s wakeup immediately after receiving the semi-passive wakeup signal and listen a maximal period of $t_c$ and sleep again if they receive nothing in that time frame.

\begin{figure}
\centering
\includegraphics[width=0.9\columnwidth]{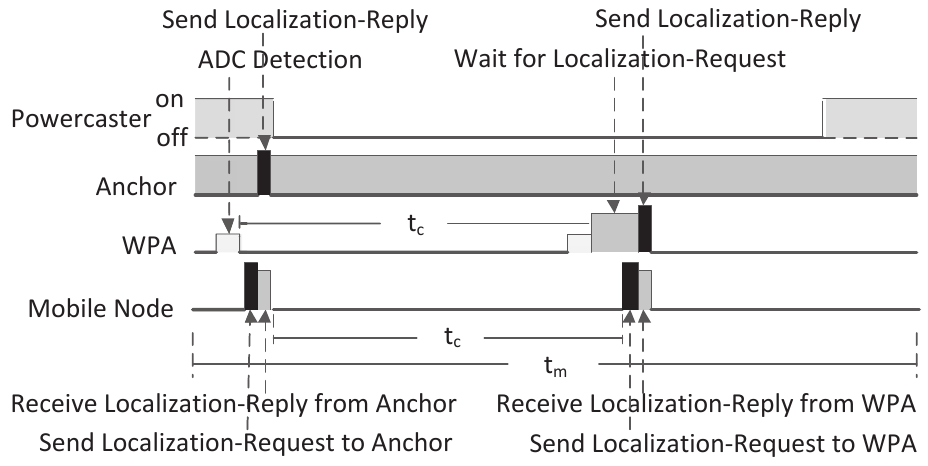}
\caption{Semi-passive wakeup process during localization period $t_{m}$.}  
\label{pic:passive_wakeup}
\end{figure}

\subsubsection{Optimization}
Although the power consumption of ADC is much lower than receiving according to the measurement results in Table~\ref{tab:power_mobile_node}, for efficiently using the limited WPT energy we optimize the ADC period $t_{\text{c}}$ to further decrease the total power consumption of {\anbe}.
Assume that $i\in \mathbf{N},\left| \mathbf{N} \right| = n$ {\anbe}s are deployed around the semi-passive wakeup range of {\charbe}. 
The time that the semi-passive wakeup signal is detected by the ADC measurement at {\anbe} $i$ is denoted as $t^{i}_{a}$. 
After sending \texttt{passive-wakeup-request} to {\charbe}, 
the {\mono} sleeps for $t_{c}$ and then sends \texttt{localization-request}. 
Then each {\anbe}s are waken up and listen for $t_{\text{rx}}^i = {t_c} - t_{a}^{i}$. 
To simplify the analysis we assume that $t_a^i$ are independent and uniformly distributed in $t_c$ with CDF $F(t_a^i) = {\textstyle{t_a^i \over {{t_c}}}},t_a^i \in [0,{t_c})$. 

\begin{proposition}

The value of $t_{c}$ producing minimum expected power consumption at WPA, $P_{a}$, is 

\begin{equation}
\arg \min_{t_c}{\left[ {{\mathrm{E}(P_a)}} \right]} \approx \sqrt {{\textstyle{{2{t_\text{m}}{t_{\text{adc}}}({p_{\text{adc}}} - {p_{\text{wfi}}})} \over {{p_{\text{rx}}} - {p_{\text{wfi}}}}}}}.
\label{eq:min_value_para}
\end{equation}
\end{proposition}

\begin{proof}
The average power consumption of a {\anbe} during a localization period $t_{m}$ is 
\begin{equation}
P_{\text{a}} = \frac{P_{\text{rx}}t_{\text{rx}}+P_{\text{tx}}t_{\text{tx}}+k_{\text{adc}}P_{\text{adc}}t_{\text{adc}}+P_{\text{wfi}}t_{\text{wfi}}}{t_m},
\label{eq:power_consumption}
\end{equation}
where the expected time of waiting for interruption during $t_{m}$ is $t_{\text{wfi}} =  t_m - (k_{\text{adc}}t_{\text{adc}} + t_{\text{rx}} + t_{\text{tx}})$, where ${k_{\text{adc}}} = \left\lfloor {{\textstyle{{{t_m}} \over {{t_c}}}}} \right\rfloor$ denotes the the number of ADC measurement during $t_m$ and $t_{\text{adc}}$, $t_{\text{rx}}$ and $t_{\text{tx}}$ denote time spent in ADC measurement, packet reception and transmission, respectively. 
To calculate $t_{\text{rx}}$ recall that the expectation of $t_a^i$ is $\mathrm{E}(t_a^i) = {\textstyle{{{t_c}} \over 2}}$.
Then the expected waiting time is $t_{\text{rx}} = t_{c} - \mathrm{E}(t_a^i) = {\textstyle{{{t_c}} \over 2}}$.

As $k_{\text{adc}}$ is discreet we replace it with a continuous value ${k^{c}_{\text{adc}}} = {\textstyle{{{t_m}} \over {{t_c}}}}$ for estimating the range of the minimum value $\mathrm{E}(P_{a})$. 
The value of $t_{c}$ that minimizes the expectation of $P_{a}$ can be calculated as 
$\arg\min_{t_c}{\left[ {\mathrm{E}({P_a})} \right]} = \{ {t_c}\left| {{\textstyle{{\partial {P_a}} \over {\partial {t_c}}}} = 0,\ {k_{\text{adc}}} = k_{\text{adc}}^{\text{c}}} \right.\}$ which results in (\ref{eq:min_value_para}). 
\end{proof}

\subsection{Range Estimation using RSS of WPT Signal}
\label{sec:distance_estimation}

The purpose of this component is to narrow down the possible location of the \mono, and restrict the number of {\anbe}s used for collision-based localization. We shall now justify the selection of the range estimation process.

\subsubsection{Reason for Range Estimation using WPT RRS}
We select range estimation using WPT RSS for two reasons: (i) the charging radio of WPT covers the localization area already, which does not require any additional communication component; (ii) Fig.~\ref{pic:rss_distance} illustrates the voltage measurement results using ADC (at $D_{\text{out}}$ pin) for various distances. The figure also illustrates that the harvested power RSS is in much more stable then the RSS of the nRF51822 at the same distance. We take advantage the attenuation of voltage over distances to estimate whether the {\mono} is inside or outside a requested range.

\subsubsection{Range Estimation Process} 

The RSS-based range estimation works as follows. Following the same deployment model as in Section~\ref{sec:passive_wakeup}, we define the cell of {\anbe} $i$ as ${\Gamma _i}$, and the voltage measurement results using ADC of a {\mono} and {\anbe}s as $\xi$ and $\mu_{i}$, respectively.
The threshold RSS value, $\theta$, is used to classify the cells of {\anbe}s. 
The cells of {\anbe}s are then categorized as (i) $\Delta_{c}$ = $\{ {\Gamma _i}\ |\ \mu_{i}  \ge \theta \}$, (ii) $\Delta_{f}$ = $\{ {\Gamma _i}\ |\ \mu_{i}  < \theta \}$, and (iii) $\overline \Delta   = \left\{ {{\Delta _f}\ |\ \xi  \ge \theta } \right\} \text{or} \left\{ {{\Delta _c}\ |\ \xi  < \theta } \right\}$.

The {\mono} broadcasts a \texttt{localization-request} with $\xi$. 
If an {\charbe} receives a \texttt{localization-request} it will compare $\xi$ with the threshold value $\theta$. 
If $\xi \leq \theta$, then the Anchor Node sends a semi-passive wakeup signal to {\anbe}s in $\Delta_{f}$; otherwise, it wakes up {\anbe}s in $\Delta_{c}$ for localization.  

If the {\mono} does receive a \texttt{localization-reply} from the WPAs, it will then request the {\charbe} to wake up the {\anbe}s in other area for cell-level localization. 
The detailed process is presented in Algorithm~\ref{alg:wiploc_localization_protocol}. 

\subsubsection{Discussion}

We need to point to certain limitations of the \name++ localization.

\begin{enumerate}
\item The value $\theta$ can calculated by $(\sum\nolimits_{i = 1}^n {{\mu _i}} )n^{-1}$. 
To simplify the implementation we use the pre-measured RSS $\theta=3.7$\,dBm at the geographic middle position of the deployment area of {\anbe}s around the {\charbe}. Such approach however might not be always practical.

\item Values of $\xi$ and $\mu$ can be affected by obstacles between powercaster and {\anbe}\,/\,{\mono}, relative direction of antennas between charger and harvester, etc. Therefore WPT RSS-based range estimation can only be used for estimating the {\mono} location with a coarse resolution.

\item The transmission power of {\anbe} messages is set to cover only the cell area which itself belongs to. 
If the {\mono} is not inside the estimated range using $\xi$, it may not receive a \texttt{localization-reply} from {\anbe}s belonging to the estimated range. 
To compensate this case, {\charbe} wakes up the {\anbe}s in the other area for the next round of localization. 
\end{enumerate}

\begin{figure}
 \centering
 \includegraphics[width=\columnwidth]{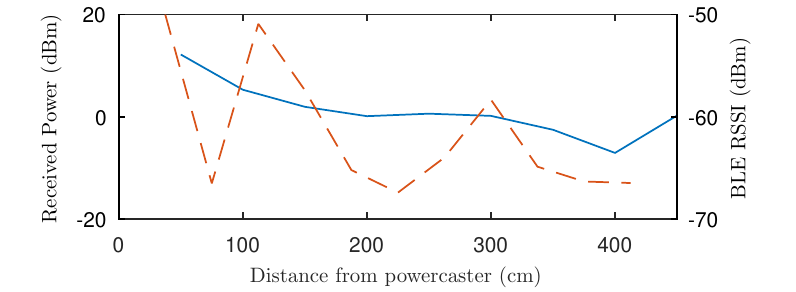}
 \caption{Signal strength of WPT and BLE communication signal. The solid line shows the received power calculated from the $D_{\text{out}}$ from the harvester. The dotted line shows the RSSI from the nRF51822.}  
 \label{pic:rss_distance}
\end{figure}

\begin{algorithm}[!t]
\footnotesize
\caption{\small Cell-level localization protocol}
\scriptsize
{\charbe} (Fixed Powered Anchor Node):
\begin{algorithmic} [1]
\Loop \label{alg2:1}
	\If{\texttt{location-request} received}
		\State Send \texttt{location-reply}
		\State Switch off and on powercaster as \texttt{passive-wakeup} signal
		\If {$\xi \leq \theta$}
			\State Broadcast \texttt{sleep} command for {\anbe}s in $\Delta_{c}$
		\Else
			\State Broadcast \texttt{sleep} command for {\anbe}s in $\Delta_{f}$
		\EndIf
	\EndIf
	\If{\texttt{location-request} with \texttt{no-reply} received}
		\State Switch off and on powercaster as \texttt{passive-wakeup} signal
		\State Broadcast \texttt{sleep} command for {\anbe}s in $\overline \Delta$
	\EndIf
\EndLoop 
\end{algorithmic}
{\anbe} (Wireless Powered Anchor Node):
\begin{algorithmic} [1]
\Loop \label{alg2:3}
	\State Sleep
	\State Monitor $D_{\text{out}}$ every $t_c$
	\If{\texttt{passive-wakeup} detected}
		\State Start receiving
		\If{\texttt{Sleep} command received}
			\State Goto Sleep
		\EndIf
		\If{\texttt{location-request} received}
			\State Send \texttt{location-reply}
		\EndIf
	\EndIf
\EndLoop 
\end{algorithmic}
{\mono}:
\begin{algorithmic} [1]
\Loop \label{alg2:2}
	\State Sleep
	\State Monitor $D_{\text{out}}$ every $t_c$.
	\If{timer $\geq t_m$}
		\State Broadcast \texttt{location-request}
		\If{\texttt{location-reply} received}
			\State Wait for $t_c$
			\State Broadcast \texttt{location-request} to WPAs
			\If{\texttt{location-reply} received from WPAs}
				\State Decode \texttt{location-reply}s and compute Room and Cell location.
			\Else
				\State Decode room-level \texttt{location-reply}
				\State Broadcast \texttt{location-request} with \texttt{no-reply} in next round
			\EndIf
		\EndIf
	\EndIf
\EndLoop 
\end{algorithmic}
\label{alg:wiploc_localization_protocol}
\end{algorithm}

\section{\name++: Implementation and Evaluation}
\label{sec:wiplocplus_implementation_evaluation}

We evaluate the overall performance of WipLoc++ through (i) PRR and localization accuracy (as in the case of \name) and (ii) the recharging period of \anbe~using harvested energy.

\subsection{Hardware Implementation}

To implement \name++ using existing \name~hardware some modifications need to be made. We list them below.

\begin{itemize}
\item{\textbf{Anchor Node:}} This node the nRF51822 is combined with the powercaster. A transistor is added to the power line of the powercaster so that the nRF51822 can control if the powercaster is on or off.
\item{\textbf{Mobile Node:}} There are no modification to the hardware of the mobile node but two extra connections are made between the nRF51822 and the harverster. 
The $D_{\text{out}}$ and the $D_{\text{set}}$ are connected to P0.01 and P0.02 of the nRF51822, respectively.
\item{\textbf{Wirelessly-Powered Anchor Node:}} This node uses the same combination of hardware and connections as the Mobile Node except for one difference. The development board for the nRF51822 is the Smart Beacon Kit and not the PCA10005.
\end{itemize}
All electrical connections of three nodes are given in Fig.~\ref{fig:schematic_complete}.

\subsection{Experiment Setup}

Due to the limited number of harvesters in our laboratory, we only deploy four {\anbe}s around one \charbe. 
The cell-level localization experiments are performed in a room and corridor separately. 
The deployment of {\anbe}s, {\charbe}s in the rooms and the testing positions of {\mono} are illustrated in Fig.~\ref{fig:experiment_setup}. 
The {\mono} sends 20 localization requests with period 1.0\,s at each testing position. As the length of the corridor is larger than the effective charging range from {\charbe} to {\anbe} we deploy two powercasters back-to-back in the middle of the corridor pointing to the begin and the end of the corridor, respectively.

\subsection{Experiment Results}

\begin{table}
	\centering
	\caption{Experiment results of {\name}++ in the room and the corridor}
	\label{tab:experiment_cell_level}
	\begin{tabular}{ c | c | c  p{\columnwidth}}
		 & PRR (\%) & Accuracy (\%)\\
		\hline
		\hline
		Room & 97.5 & 59.9 \\
		Corridor & 100 & 82.2 \\
	\end{tabular}
\end{table}

\begin{figure}
 \label{pic:cell_result_office_corridor}
 \centering	
 \subfigure[t][Localization accuracy: room]{
 	\includegraphics[width=0.5\columnwidth]{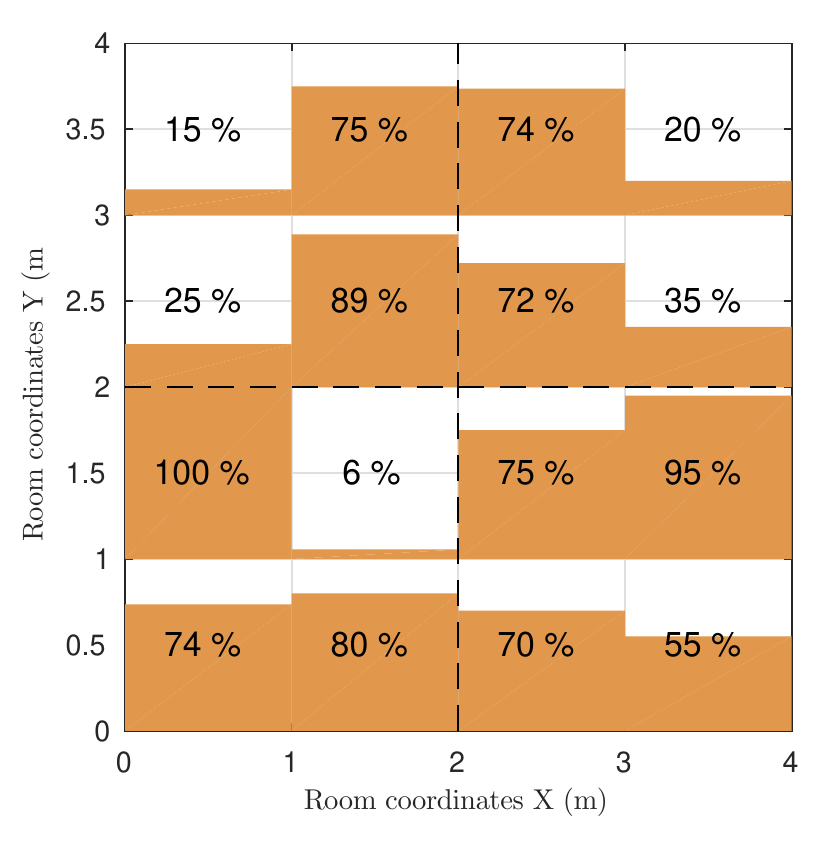}
 	\label{pic:cell_result_office}
 }
 \hspace{0.4cm}
 \subfigure[t][Localization accuracy: corridor]{
 	\includegraphics[height=4.5cm]{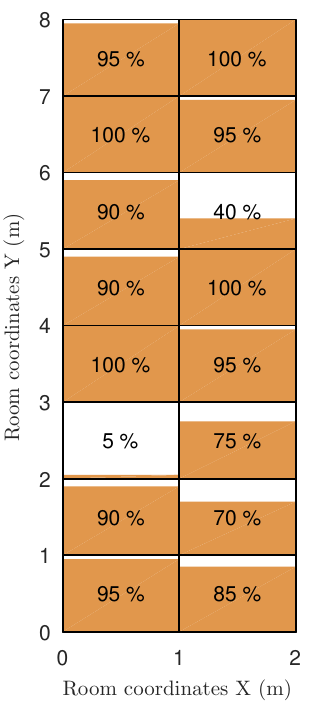}
  \label{pic:cell_result_corridor}
 }
 \caption{\name++ localization results. The values in orange rectangles of two figures represent the cell-level localization accuracy of {\name}++ measured per testing positions of (a) room two and (b) corridor as shown in Fig.~\ref{fig:experiment_setup}.}
\end{figure}

Compared with the \name~room-level results, (see Table~\ref{tab:result1}), the average cell-level PRR and localization accuracy of {\name}++ in the room and corridor (see Table~\ref{tab:experiment_cell_level}) are on the same level. 
Although the cell-level localization accuracy in the room is lower than the room-level accuracy, the localization cell is only 4\,$\text{m}^\text{2}$ using {\name}++, which is much smaller than the 16\,$\text{m}^\text{2}$ localization room using {\name}.
The test results prove that {\name}++ is able to achieve cell-level localization using WPT in all deployed nodes except the {\charbe}.

On the other hand, we find that cell-level localization accuracy at some positions in Fig.~\ref{pic:cell_result_office_corridor} is much lower than the average value for the whole area. 
This is due to two reasons: (i) the radio pattern of powercaster has only 60$^\circ$ coverage in width and height, therefore some testing positions at the border of the room are not effectively covered and (ii) the threshold value $\theta$ uses the measured RSS value at the middle position of the charging area. 
However, the contour line of $\theta$ is not straight in the radio pattern of powercaster, which makes it difficult to categorize cells into strict squares.

\section{Conclusion}
\label{sec:conclusion}

In this paper we presented an RF WPT-enabled indoor localization system denoted as \textbf{\name} (\textbf{Wi}reless \textbf{P}owered \textbf{Loc}alization system). 
The key innovations of WipLoc include: 
(i) leveraging collisions and orthogonal codes to build an extremely low power localization approach, 
and (ii) constructing a cell-level localization network by managing the limited harvested energy from RF-based WPT systems.  
Based on extensive indoor experiments, we showed that \name~is capable of providing continuous cell-level localization to mobile nodes.
To the best of our knowledge, \name~is the first localization system powered by RF transmission.  

\bibliographystyle{IEEEtran}
\bibliography{IEEEabrv,wireless_power_localization}

\end{document}